\documentclass[conference]{IEEEtran}
\IEEEoverridecommandlockouts

\usepackage{ifpdf}
\ifCLASSINFOpdf
  \usepackage[pdftex]{graphicx}
\else
\fi
\usepackage{amsmath}
\usepackage{amsfonts}
\usepackage{amssymb}
\usepackage{amsthm}
\usepackage{algorithmic}
\usepackage{array}
\usepackage{float}
\usepackage{url}
\newtheorem{theorem}{Theorem}
\newtheorem{definition}{Definition}
\newtheorem{example}{Example}

\begin{document}
\title{Lattice Construction $C^\star$ from Self-Dual Codes}


\author{\IEEEauthorblockN{Maiara F. Bollauf \IEEEauthorrefmark{1}, Sueli I. R.  Costa\IEEEauthorrefmark{1}, and 
Ram Zamir \IEEEauthorrefmark{3}}
\vspace{0.2cm}
\IEEEauthorblockA{\IEEEauthorrefmark{1} Institute of Mathematics, Statistic and Scientific Computing \\
University of Campinas, S\~ao Paulo\\
13083-859, Brazil \\ Email: bollauf@ieee.org, sueli@ime.unicamp.br}
\vspace{0.2cm}
\IEEEauthorblockA{\IEEEauthorrefmark{2} Deptartment Electrical Engineering-Systems \\
Tel Aviv University, Tel Aviv, Israel \\ Email: zamir@eng.tau.ac.il}}
\maketitle

\begin{abstract} 

	Construction $C^\star$ was recently introduced as a generalization of the multilevel Construction C (or Forney's code-formula), such that the coded levels may be dependent. Both constructions do not produce a lattice in general, hence the central idea of this paper is to present a $3-$level lattice Construction $C^\star$ scheme that admits an efficient nearest-neighborhood decoding. In order to achieve this objective, we choose coupled codes for levels $1$ and $3,$ and set the second level code $\mathcal{C}_2$ as an independent linear binary self-dual code, which is known to have a rich mathematical structure among families of linear codes. Our main result states a necessary and sufficient condition for this construction to generate a lattice. We then present examples of efficient lattices and also non-lattice constellations with good packing properties. \\

\end{abstract}



% \fi
\IEEEpeerreviewmaketitle

\vspace{0.1cm}

\section{Introduction}

%This demo file is intended to serve as a ``starter file''
%for IEEE conference papers produced under \LaTeX\ using
%IEEEtran.cls version 1.8b and later.

 

	A lattice is a well studied mathematical structure due to an extensive list of applications, including its efficient packing properties. The sphere packing problem has known solutions only for dimensions $2,3,8$ and $24,$ \cite{Hales:2005, Cohn2016,Viazovska2016} and all of them can be reached by lattices. For other dimensions, there are strong beliefs that the best possible packing density can be achieved by lattices.
	
	One way of producing lattice constellations is to use linear codes in the so called Constructions A, B, and D \cite{conwaysloane}. There are also other interesting constructions that generate more general constellations (lattices and non-lattices) with prominent applications in quantization and coded modulation, such as Constructions C \cite{forney1} and $C^\star$~\cite{bzc2019}. The advantage of working with such constructions is mainly the translation of characteristics from the linear code over a finite field to an infinite constellation in the $n-$dimensional real space. %with a special addition property.
	
	While the condition for Construction C to be a lattice is elegant and directly related to Construction D \cite{kositoggier}, the lattice condition for its generalization, i.e. Construction $C^\star,$ cannot be related to any other previous lattice construction \cite{bzc2019}. Thus, one proposal of this work is to investigate families of codes which make Construction $C^\star$ always a lattice and the result points out to the role of self-dual codes.
	
	In coding theory, self-dual codes are of a peculiar importance as they represent the best known error correcting codes for transmission or data storage \cite{Huffman:2005}, when one is interested in transmitting a large number of messages with a large minimum weight, in order to correct maximum number of errors. Their properties and relations with results from group theory, combinatorics and lattices are well known. Self-dual codes underlying Construction A are explored in several works \cite{bachoc96, Nebe2006, Sloane:1979} regarding the association of these codes to unimodular lattices.
	
	We are inspired by the $3-$level Construction $C^\star$ of the Leech lattice presented in \cite{bzc2019}, which considered coupled codes for levels $1$ and $3,$ while the second level was the $[24,12,8]-$Golay code. We generalize this idea for any even dimension by fixing the choice of the second level code $\mathcal{C}_2$ to be a self-dual code and our main result states a necessary and sufficient condition for such construction to produce a lattice. This theory also arises as a promising approach for the open problem of decoding Construction $C^\star,$ by using an extension of the works from Forney \cite{forney2} and Amrani et al. \cite{amrani} to any $3-$level lattice Construction $C^\star.$
	
%	which considered $\mathcal{C}_1$ as the repetition code, $\mathcal{C}_2$ as the extended Golay code $\mathcal{C}_{24}$ and $\mathcal{C}_3$ as $\mathbb{F}_2^{24},$ written as a special union of sets composed by even and odd codewords respectively, we generalize this idea for any even dimension $n,$ by fixing the choice of $\mathcal{C}_2$ to be any self-orthogonal (in our examples mainly self-dual) code. This theory also arises as a promising approach for the open problem of decoding of Construction $C^\star,$ by using an extension of the work of Forney \cite{forney2} and Amrani et al. \cite{amrani} to any $3-$level lattice Construction $C^\star.$
	
	We present alternative constructions for the $E_8$ lattice and known packings in dimension $32$ and $40.$ Interesting non-lattice constellations (with a code $\mathcal{C}_2$ which is not self-dual), including a special one in dimension 4 that achieves the same packing density of the lattice $D_4,$ are presented.
	
	This paper is organized as follows: Section~\ref{Sec:2} introduces some relevant notions about lattices, Construction $C^\star,$ and codes. Section~\ref{Sec:3-level} presents a general way of producing lattices via a $3-$level Construction $C^\star$ by using self-dual codes in the second level. Section~\ref{Sec:4} is devoted to examples of lattice packings. Section~\ref{Sec:5} describes non-lattice constellations which have good packing properties, including one that presents the same packing density as the densest known lattice in $\mathbb{R}^4.$ Finally, in Section~\ref{Sec:Conc} conclusions and perspectives of future work are drawn.

\vspace{0.2cm}

\section{Background on lattices and codes} \label{Sec:2}
	
	In this section, we recall the definition of Construction $C^\star$ and the condition for it to be a lattice. We also point out some properties of self-orthogonal and self-dual codes. % and relate them to a general lattice $C^\star$ construction. %We denote by $+$ the real addition and by $\oplus$ the sum in $\mathbb{F}_{2},$ i.e., $x \oplus y=(x+y)mod \ 2.$
	
%\begin{definition}(Linear binary code) A linear binary code $\mathcal{C}$ of length $n$ and rank $k$ is a linear subspace with dimension $k$ over the vector space $\mathbb{F}_{2}^{n},$ with $2^{k}$ codewords. 
%\end{definition}	
%
%\begin{definition}(\textit{Generator matrix of a binary code $\mathcal{C}$}) A linear binary code $\mathcal{C}$ can be written as the image of an injective linear transformation 
%\begin{eqnarray*} 
%\phi &: & \mathbb{F}_{2}^{k} \rightarrow \mathbb{F}_{2}^{n} \nonumber \\
%(a_{1}, \dots, a_{k}) &\mapsto & G \cdot (a_{1}, \dots, a_{k})^{T},
%\end{eqnarray*}
%where $G \in \mathbb{Z}_{2}^{n \times k}$ is the generator matrix of $\mathcal{C}.$
%\end{definition}
%
%\begin{definition}(\textit{Parity check matrix}) A parity check matrix $H \in \mathbb{F}_{2}^{n-k,n}$ of a linear binary code $\mathcal{C}$ is a matrix with a property to detect if an element $c \in \mathbb{F}_{2}^{n}$ is a codeword of $\mathcal{C},$ i.e.,
%\begin{equation}
%Hc^{T}=0  \in \mathbb{F}_{2}^{n-k} \Leftrightarrow c \in \mathcal{C} \subset \mathbb{F}_{2}^{n}.
%\end{equation} 
%\end{definition}

\begin{definition}(Lattice) A lattice $\Lambda \subset \mathbb{R}^{N}$ is a set of integer linear combinations of independent vectors $v_{1}, v_{2}, \dots, v_{n} \in \mathbb{R}^{N}.$ %In other words, we can also say that a lattice is a discrete additive subgroup of $\mathbb{R}^{n}.$		
\end{definition}

	We say that a lattice is full rank if $N=n,$ which is the case of lattices explored through this paper. The volume $vol(\Lambda)$ of a full rank lattice is the absolute value of the determinant of a matrix which has its columns as the generator vectors $v_{1}, v_{2}, \dots, v_{n}.$
	
\begin{definition}(\textit{Packing radius and packing density}) The packing radius $r_{pack}(\Lambda)$ of a lattice $\Lambda \subset \mathbb{R}^n$ is half of the minimum distance between lattice points and the packing density $\Delta(\Lambda)$ is the fraction of the space that is covered by balls $B(\lambda, r_{pack}(\Lambda))$ of radius $r_{pack}(\Lambda),$ centered at a lattice point $\lambda \in \Lambda,$ i.e.,
\begin{equation}
\Delta(\Lambda) = \dfrac{vol (B(0,r_{pack}(\Lambda))}{vol(\Lambda)} = \dfrac{V_n~ r_{pack}^n}{vol(\Lambda)},
\end{equation}
where $V_n$ refers to the volume of the unit ball in $\mathbb{R}^n.$ 
\end{definition}

	The packing density is an important measure to compare lattices. However, for increasing dimensions, this value tends to zero and analogies are hard to perform. In that case, instead of analyzing packing densities it is common to compare Hermite constants.
	
\begin{definition}(\textit{Hermite constant}) The Hermite constant of a lattice $\Lambda \subset \mathbb{R}^n$ is given by
\begin{equation}
\gamma_n(\Lambda) = 4 \left( \dfrac{\Delta(\Lambda)}{V_n} \right)^{2/n} = \dfrac{4 r_{pack}^2}{vol(\Lambda)^{2/n}} = \dfrac{d_{\min}^2(\Lambda)}{vol(\Lambda)^{2/n}},
\end{equation}
where $V_n$ refers to the volume of the unit ball in $\mathbb{R}^n.$ 
\end{definition}

	The Hermite constant $\gamma_n$ measures the highest attainable coding gain of an $n-$dimensional lattice.
	
	Besides the well known Constructions A and D, that produce lattice constellations from linear codes, another interesting construction is the so called Construction C or construction by code-formula \cite{forney1}.
	
\begin{definition}(\textit{Construction C})  Consider $L$ binary codes $\mathcal{C}_{1}, \dots, \mathcal{C}_{L} \subseteq \mathbb{F}_{2}^{n},$ not necessarily nested or linear. Then we define an infinite constellation $\Gamma_{C}$ in $\mathbb{R}^{n}$ that is called Construction C as:
	\begin{equation} \label{eqC}
	\Gamma_{C}:=\mathcal{C}_{1}+2\mathcal{C}_{2}+ \dots + 2^{L-1}\mathcal{C}_{L}+2^{L}\mathbb{Z}^{n}.
	\end{equation} 
\end{definition}

	A generalization of Construction C was introduced in \cite{bzc2018, bzc2019} and denoted by Construction $C^\star.$ It was inspired by bit-interleaved coded modulation (BICM) and asymptotically, it was demonstrated its superior packing efficiency when compared to Construction C.
	
	The main feature of Construction $C^\star$ that differs from Construction C is the fact that the levels are inter-coded, i.e., they are dependent.

\begin{definition} \label{constrcstar} (\textit{Construction $C^{\star}$}) Let $\mathcal{C} \subseteq \mathbb{F}_{2}^{nL}$ be a binary code. Then Construction $C^{\star}$ is defined as 
{\small \begin{eqnarray}
 \Gamma_{C^{\star}}& := & \{c_{1}+2c_{2}+ \dots + 2^{L-1}c_{L}+2^{L}z: (c_{1}, c_{2}, \dots, c_{L}) \in \mathcal{C}, \nonumber \\
& & c_{i} \in \mathbb{F}_{2}^{n}, i=1, \dots, L, z \in \mathbb{Z}^{n}\}.  
\end{eqnarray}}
\end{definition}

	Note that Construction C coincides with Construction $C^\star$ when $\mathcal{C}=\mathcal{C}_1 \times \dots \times \mathcal{C}_n$ and we observe that both constructions in general do not produce a lattice. A condition that will assure the laticeness of Construction $C^\star$ will be presented next.
	
\begin{definition} (\textit{Projection codes}) \label{subcodes} Let $c=(c_1,...,c_L)$ be a partition of a codeword $c = (c_{11}, \dots, c_{1n},....,c_{L1}, \dots, c_{Ln}) \in \mathcal{C} \subseteq \mathbb{F}_2^{nL}$ into length$-n$ subvectors  $c_i = (c_{i1},....,c_{in}),$  $i=1,\dots,L.$ Then, a projection code $\mathcal{C}_i$ consists of all subvectors $c_{i}$ that appear as we scan through all possible codewords $c \in \mathcal{C}.$ %Note that if $\mathcal{C}$ is linear, then every projection code $\mathcal{C}_{i}, i=1, \dots, L$ is also linear.
\end{definition}

	In what follows, we denote by $+$ the real addition and by $\oplus$ the sum in $\mathbb{F}_{2},$ i.e., $x \oplus y=(x+y) \bmod \ 2.$

\begin{definition} (\textit{Antiprojection}) The antiprojection $\mathcal
{S}_{i}(c_1,\dots,$ $c_{i-1}, c_{i+1}, \dots, c_{L})$ consists of all vectors $c_{i} \in \mathcal{C}_{i},$ $i=1,\dots, L$ that appear as we scan through all possible codewords $c \in \mathcal{C},$ while keeping $c_{1}, \dots, c_{i-1}, c_{i+1}, \dots, c_{L}$ fixed:
\begin{align}
\mathcal{S}_{i}(c_1,...,c_{i-1}, c_{i+1},...,c_{L}) & = \{c_{i} \in \mathcal{C}_{i}: \nonumber \\
& \hspace{-0.5cm} (c_{1}, \dots, \underbrace{c_{i}}_{\text{\makebox[0pt]{i-th position} }}, \dots, c_{L}) \in \mathcal{C}\}.
\end{align}
\end{definition}

	In \cite{bzc2019}, there are two statements that guarantee the latticeness of Construction $C^\star$ and here we recall one of them, due to its simplicity and straightforward relation with the results of this paper. We start by the definition of Schur product.
	
\begin{definition} (\textit{Schur product}) For $x=(x_{1}, \dots, x_{n})$ and $y=(y_{1}, \dots, y_{n})$ both in $\mathbb{F}_{2}^{n},$ we define $x \ast y = (x_{1}y_{1}, \dots, x_{n}y_{n}).$
\end{definition}

	Consider $\psi: \mathbb{F}_2^{n} \rightarrow \mathbb{R}^{n}$ as the natural embedding. Then, for $x, y \in \mathbb{F}_{2}^{n},$ it is valid that
\begin{equation}\label{eqss}
\psi(x)+\psi(y)=\psi(x \oplus y)+2\psi(x \ast y).
\end{equation} 
In order to simplify, we abuse the notation, writing Eq. \eqref{eqss} as	
\begin{equation} \label{sum}
x+y=x \oplus y + 2(x \ast y).
\end{equation}
	
	A chain $\mathcal{C}_1 \subseteq \mathcal{C}_2 \subseteq \mathbb{F}_2^n$ is said to be closed under Schur product if for any $c_1, \tilde{c}_1 \in \mathcal{C}_1,$ the Schur product $c_1 \ast \tilde{c}_1 \in \mathcal{C}_2.$

\begin{theorem} \cite{bzc2019} (\textit{A sufficient lattice condition for $\Gamma_{C^\star}$})  \label{coro4} If $\mathcal{C} \subseteq \mathbb{F}_{2}^{nL}$ is a linear binary code with projection codes $\mathcal{C}_{1},\mathcal{C}_{2}, \dots, \mathcal{C}_{L}$ such that $\mathcal{C}_{1} \subseteq \mathcal{S}_{2}(0,\dots,0) \subseteq \mathcal{C}_{2} \subseteq \dots \subseteq \mathcal{C}_{L-1} \subseteq \mathcal{S}_{L}(0, \dots,0) \subseteq \mathcal{C}_{L} \subseteq \mathbb{F}_{2}^{n}$ and the chain $\mathcal{C}_{i-1} \subseteq \mathcal{S}_{i}(0,\dots,0)$ is closed under the Schur product for all $i=2,\dots, L,$ then $\Gamma_{C^{\star}}$ is a lattice.
\end{theorem}

	In this paper we set $L=3$ for Construction $C^\star$ and analyze the case where the second level code $\mathcal{C}_2$ is a self-orthogonal linear code in $\mathbb{F}_2^n,$ independent of the other two levels. In $\mathbb{F}_2,$ the standard inner product of $c=(c_1, c_2, \dots, c_n)$ and $\tilde{c}=(\tilde{c}_1, \tilde{c}_2, \dots, \tilde{c}_n)$ is defined as $\langle c, \tilde{c} \rangle = \sum_{i=1}^{n} c_i \tilde{c}_i \bmod 2$ and the orthogonal set $\mathcal{C}^\perp$ of a code $\mathcal{C} \subseteq \mathbb{F}_2^n$ is also defined as the set $\mathcal{C}^\perp=\{c \in \mathbb{F}_2^n: \langle c,\tilde{c} \rangle = 0,~ \forall \tilde{c} \in \mathcal{C}\}.$
	
\begin{definition}\label{defself} (\textit{Self-orthogonal and self-dual codes}) A code $\mathcal{C}$ is self-orthogonal if $\mathcal{C} \subset \mathcal{C}^{\perp}$ and  it is self-dual if $\mathcal{C} = \mathcal{C}^{\perp}.$
\end{definition}

%	This definition is invariant over the alphabet where the code is defined and it our case we will assume only binary self-orthogonal (resp. self-dual) codes.

	A code $\mathcal{C}$ is self-orthogonal if and only if $\langle c, \tilde{c} \rangle = 0,$ for all $c, \tilde{c} \in \mathcal{C}.$ Each codeword in a  self-orthogonal code has even Hamming weight and $(1,\dots, 1) \in \mathcal{C}^\perp.$ Indeed, let $c \in \mathcal{C},$ which is a self-orthogonal code, then $\langle c,c \rangle = 0,$ and it means that the Hamming weight of $c,$ i.e. $\omega(c)$, is always even for all $c \in \mathcal{C}.$ Also, $(1,\dots, 1) \in \mathcal{C}^{\perp}$ due to the fact that $\langle c, (1,\dots, 1) \rangle = 0,$ for all $c \in \mathcal{C}$ and $\omega(c)$ is even.

%	It is not hard to see that Definition \ref{defself} implies that a code $\mathcal{C}$ is self-orthogonal if and only if $\langle c, \tilde{c} \rangle = 0,$ for all $c, \tilde{c} \in \mathcal{C}.$ 

%	Regarding the codewords of a $[n,k,d]-$ binary linear self-orthogonal code, it is valid that they all of them have even weight and $\mathcal{C}^{\perp}$ contains the codeword $(1,\dots, 1).$ Indeed, let $c \in \mathcal{C},$ which is a self-orthogonal code, then $\langle c,c \rangle = 0,$ and it means that the Hamming weight of $c,$ i.e. $\omega(c)$, is always even for all $c \in \mathcal{C}.$ Also, $(1,\dots, 1) \in \mathcal{C}^{\perp}$ due to the fact that $\langle c, (1,\dots, 1) \rangle = 0,$ for all $c \in \mathcal{C}$ and $\omega(c)$ is even.
	
%\begin{theorem}\label{thmself}\cite[p. 28]{huffpless} (Properties of binary self-orthogonal codes) Let $\mathcal{C}$ be an $[n,k,d]-$ linear code over $\mathbb{F}_2.$ Then each codeword of $\mathcal{C}$ has even weight and $\mathcal{C}^{\perp}$ contains the codeword $(1,\dots, 1).$
%\end{theorem}

%\begin{proof} Let $c \in \mathcal{C},$ which is a self-orthogonal code, then $\langle c,c \rangle = 0,$ and it means that the Hamming weight of $c,$ i.e. $\omega(c)$, is always even for all $c \in \mathcal{C}.$ Also, $(1,\dots, 1) \in \mathcal{C}^{\perp}$ due to the fact that $\langle c, (1,\dots, 1) \rangle = 0,$ for all $c \in \mathcal{C}$ and $\omega(c)$ is even.
%\end{proof}

%	To particularly identify self-dual codes, there is a simple condition discussed in \cite{ebeling}.

	A characterization of self-dual codes is given by \cite[p. 8]{ebeling}\cite{MacWilliams1977}: a $[n,k,d]-$ linear code $\mathcal{C}$ is self-dual if and only if $\mathcal{C} \subset \mathcal{C}^\perp$ and $k=\tfrac{n}{2}.$
	
%\begin{proposition}\cite[p. 9]{ebeling} (Relation between self-orthogonal and self-dual codes) A $[n,k,d]-$ linear code $\mathcal{C}$ is self-dual if and only if it is self-orthogonal, i.e., $\mathcal{C} \subseteq \mathcal{C}^\perp$ and $k=\tfrac{n}{2}.$
%\end{proposition}

\begin{example} The Reed-Muller code $\mathcal{RM}(1,4),$ which is a $[16,5,8]-$binary linear code is self-orthogonal, while the $[8,4,4]-$extended Hamming code and the $[24,12,8]-$extended Golay code are both examples of self-dual codes. 
\end{example}

\vspace{0.2cm}

\section{General lattices via 3-level Construction $C^\star$} \label{Sec:3-level}

	Inspired by the Leech lattice construction via $C^\star$ presented in \cite{bzc2018}, we aim to describe a more general $3-$level lattice Construction $C^\star$ by fixing the level (projection) codes as
	
\begin{itemize}
\item $\mathcal{C}_1 = \{(0,\dots,0), (1, \dots, 1)\}\subset \mathbb{F}_2^n,$ which is the repetition code;
\item $\mathcal{C}_2 \subset \mathbb{F}_2^n$ as a convenient code we are going to explore later;
\item $\mathcal{C}_{3} = \tilde{\mathcal{C}}_{3} \cup \overline{\mathcal{C}}_{3} = \mathbb{F}_{2}^{n},$ and we require that if $c_1=(0,\dots,0)$ then $c_3 \in \tilde{\mathcal{C}}_3=\{(x_{1}, \dots, x_{n}) \in \mathbb{F}_{2}^{n}:  \sum_{i=1}^{n} x_{i} \equiv 0 \mod 2\}$ and if $c_1=(1,\dots,1)$ then $c_3 \in \overline{\mathcal{C}}_3=\{(y_{1}, \dots, y_{n}) \in \mathbb{F}_{2}^{n}: \sum_{i=1}^{n} y_{i} \equiv 1 \mod 2\}.$
\end{itemize}
In other words, the main code $\mathcal{C} \subseteq \mathbb{F}_{2}^{3n}$ is given by
\begin{align}\label{eq3level}
\mathcal{C}= & \{(\underbrace{0,\dots, 0}_{\in \mathcal{C}_1}, \underbrace{a_{1}, \dots, a_{n}}_{\in \mathcal{C}_{2}}, \underbrace{x_{1}, \dots, x_{n}}_{\in \tilde{\mathcal{C}}_{3}}), \nonumber \\
& (\underbrace{1, \dots, 1}_{\in \mathcal{C}_1}, \underbrace{a_{1}, \dots, a_{n}}_{\in \mathcal{C}_{2}}, \underbrace{y_{1}, \dots, y_{n}}_{\in \overline{\mathcal{C}}_{3}})\}.
\end{align}

	One can notice that the dependence between levels is crucial in the definition of the main code $\mathcal{C} \subseteq \mathbb{F}_{2}^{3n},$ as in Eq. \eqref{eq3level}. %The choice of codewords that constitute $\mathcal{C}$ rely upon the particular choice of first and third coupled codes. 
We can then define a constellation $\Gamma_{\mathcal{C}^\star}$ as the $3-$level Construction $C^{\star}$ given by
\begin{equation}\label{lattice_cstar}
\Gamma_{\mathcal{C}^\star}=\{c_{1}+2c_{2}+4c_{3}+8z: (c_{1}, c_{2}, c_{3}) \in \mathcal{C}, z \in \mathbb{Z}^{n}\}.
\end{equation}
	
	The choice of $\mathcal{C}_2$ in Eq. \eqref{eq3level} is directly related to Theorem \ref{coro4}, as we are interested in constructing lattice constellations. %We denote by $+$ the real addition and by $\oplus$ the sum in $\mathbb{F}_{2},$ i.e., $x \oplus y=(x+y) \bmod \ 2.$
	
\begin{theorem}\label{self3} (\textit{Lattice Construction $C^\star$ with self-orthogonal codes}) Let $\mathcal{C} \subset \mathbb{F}_{2}^{3n}$ be a linear code according to Eq. \eqref{eq3level}. The resulting constellation $\Gamma_{\mathcal{C}^\star}$ (Eq. \eqref{lattice_cstar}) obtained via Construction $C^\star$ from the code $\mathcal{C}$ is a lattice if and only if $\mathcal{C}_2 \subseteq \mathbb{F}_2^n$ is a self-orthogonal code that contains $(1,\dots, 1).$ 
\end{theorem}

\begin{proof}$(\Rightarrow)$ Suppose that $\Gamma_{C^\star}$ constructed from $\mathcal{C}\subseteq \mathbb{F}_{2}^{3n}$ is a lattice. Then, given $x, y \in \Gamma_{C^\star}$ it is true that $x+y \in \Gamma_{C^\star}.$ We can write
\begin{align*}
x=c_1+2c_2+4c_3+8z \\
y= c_1' + 2c_2' + 4c_3' + 8z'
\end{align*}
and $x+y \in \Gamma_{C^\star}$ implies that the vector
\begin{align}\label{eqcode}
& ~ ~ ~ ~ ~ ~ ~ ~ ~ (c_1 \oplus c_1', c_2 \oplus c_2' \oplus (c_1 \ast c_1'), \nonumber \\
&c_3 \oplus c_3' \oplus ((c_1 \ast c_1')\ast (c_2 \oplus c_2') \oplus (c_2 \ast c_2')) \in \mathcal{C}
\end{align}
and in particular, $c_2 \oplus c_2' \oplus (c_1 \ast c_1') \in \mathcal{C}_2.$ Due to linearity, $c_2 \oplus c_2' \in \mathcal{C}_2$ and for $c_1=c_1'=(1,\dots, 1),$ we must have that $(1,\dots, 1) \in \mathcal{C}_2.$

	It remains to demonstrate the $\mathcal{C}_2$ is self-orthogonal. There are only four possible choices for $c_1$ and $c_1',$ which we discuss case by case below:
	
\begin{itemize}
\item $c_1=c_1'=(0,\dots,0):$ from Eq. \eqref{eqcode} we have that $(0,\dots,0, c_2 \oplus c_2', c_3 \oplus c_3' \oplus c_2 \ast c_2') \in \mathcal{C},$ where by construction $c_3 \oplus c_3'$ has even weight, so it is straightforward to conclude that the sum of the coordinates of $c_2 \ast c_2'$ is equal to zero and $\langle c_2, c_2' \rangle =0.$
\item $c_1=(1,\dots,1)$ and $c_1'=(0,\dots,0):$ from Eq. \eqref{eqcode} we have that $(1,\dots,1, c_2 \oplus c_2', c_3 \oplus c_3' \oplus c_2 \ast c_2') \in \mathcal{C},$ where by construction the coordinates of $c_3$ sum one modulo 2 and the coordinates of $c_3'$ sum zero modulo 2, thus the only possibility is that the sum of $c_2 \ast c_2'$ is equal to zero and $\langle c_2, c_2' \rangle =0.$ An analogous argument applies to the case where $c_1=(0,\dots,0)$ and $c_1'=(1,\dots,1).$
\item $c_1=c_1'=(1,\dots,1):$ from Eq. \eqref{eqcode} we have that $(0,\dots,0, c_2 \oplus c_2' \oplus (1,\dots,1), c_3 \oplus c_3' \oplus (c_2 \oplus c_2') \oplus c_2 \ast c_2') \in \mathcal{C},$ where in this case both coordinates of $c_3$ and $c_3'$ sum one modulo 2, hence $c_3 \oplus c_3'$ has even weight and consequently also $(c_2 \oplus c_2') \oplus c_2 \ast c_2'$ must have even weight. We need to prove that the coordinates of $c_2 \ast c_2'$ sum zero modulo 2. Assume that $c_2 \oplus c_2'$ has odd weight, by contradiction (because it will force $c_2 \ast c_2'$ to have odd weight as well). Due to the linearity of $\mathcal{C}_2,$  $c_2 \oplus c_2'=\tilde{c}_2 \in \mathcal{C}_2.$ Then, we consider in Eq.~\eqref{eqcode}, $c_2=c_2'=\tilde{c}_2,$ which yields:
\begin{equation}\label{eqcontradiction}
(0,\dots,0, 1, \dots, 1, c_3 \oplus c_3' \oplus \tilde{c}_2 \oplus \tilde{c}_2 \oplus \tilde{c}_2 \ast \tilde{c}_2) \in \mathcal{C},
\end{equation}
and $(\tilde{c}_2 \oplus \tilde{c}_2) \oplus (\tilde{c}_2 \ast \tilde{c}_2) = \tilde{c}_2,$ what makes the third coordinate to have odd weight. Thus, the element written in Eq.~\eqref{eqcontradiction} does not belong to the code $\mathcal{C}$ and we have a contradiction. Therefore, both $c_2 \oplus c_2'$ and $c_2 \ast c_2'$ must have even weight, what implies that $\langle c_2, c_2' \rangle =0.$ 
\end{itemize}
	
	We can then conclude that $\mathcal{C}_2$ is self-orthogonal.	\\
	
\noindent $(\Leftarrow)$ To assure the latticeness condition from Theorem \ref{coro4} to hold one needs to first verify that
\begin{equation}
\mathcal{C}_{1} \subseteq \mathcal{S}_{2}(0,\dots,0) \subseteq \mathcal{C}_{2} \subseteq \mathcal{S}_{3}(0, \dots,0) \subseteq \mathcal{C}_{3},
\end{equation}
and due to the structure of $\mathcal{C} \subseteq \mathbb{F}_{2}^{3n}$ in Eq. \eqref{eq3level} we have that $\mathcal{S}_{2}(0,\dots,0) = \mathcal{C}_2$ and $\mathcal{S}_{3}(0, \dots,0) = \tilde{\mathcal{C}}_{3}.$ By hypothesis, $(1,\dots, 1) \in \mathcal{C}_2,$ what allow us to conclude that $\mathcal{C}_{1} \subseteq \mathcal{S}_{2}(0,\dots,0)$ and this nesting is clearly closed under Schur product.

	Since $\mathcal{C}_2$ is self-orthogonal, all codewords have even weight and $\mathcal{C}_{2} \subseteq \tilde{\mathcal{C}}_{3}.$ It remains to show that this nesting is closed under Schur product, i.e., given any $c_2, c_2' \in \mathcal{C}_2,$ the sum of all coordinates of the vector defined by $c_2 \ast c_2'$ should be zero modulo 2. Observe that the Schur product is the coordinate-by-coordinate product and the action of summing all components of the resulting Schur product vector is the same as $\langle  c_2, c_2' \rangle.$ Thus, we want to prove that $\langle  c_2, c_2' \rangle =0 \bmod 2,$ which is true since $\mathcal{C}_2$ is self-orthogonal.
\end{proof}	

	One can observe that for self-dual codes, the condition required by Theorem \ref{self3} is automatically satisfied, because $\mathcal{C}=\mathcal{C}^\perp$ and also $(1, \dots, 1) \in \mathcal{C}^\perp.$% which justifies the choice in the examples explored below. 
	
\vspace{0.2cm}
		
\section{Constructions of known lattices via $C^\star$} \label{Sec:4}

	We can only expect to have interesting lattice constellations via Construction $C^\star$ following the procedure described in Section \ref{Sec:3-level} for $n$ even, because we need to assure that $(1,\dots, 1) \in \mathcal{C}_2 \subseteq \mathcal{S}_{3}(0, \dots,0) = \tilde{\mathcal{C}}_{3}.$ 
	
	This section summarizes some new lattice constructions for even dimensions built from a $3-$level Construction $C^\star$ with the main code $\mathcal{C} \subseteq \mathbb{F}_2^{3n}$ as in Eq. \eqref{eq3level}, whose resulting constellation is $\Gamma_{C^\star}$ as in Eq. \eqref{lattice_cstar}.  
	
	Observe that an essential feature to calculate the packing efficiency or Hermite constant of a lattice is the minimum distance. A closed formula for the minimum distance of a constellation generated by Construction $C^\star$ is still an open problem and in general, what is known is just an upper and lower bound for it \cite{bzc2019}. However, for particular cases, when the codes are established, as it is the case of the examples explored in this section, this calculation can be done by brute force, i.e., by investigating all possible minimum weight codewords and calculating the minimum among them. 
	
\vspace{0.2cm}
	
\noindent \textbf{Dimension 8 - $E_{8}$ lattice:} Define $\mathcal{C}_2$ as the $[8,4,4]-$extended Hamming code, which is self-dual and whose basis vectors are displayed in the rows of the following generator matrix,
%\begin{equation}
%G=\left[ \begin{array}{@{}*{16}{c}@{}}
%1 & 0 & 0 & 0 \\
%1 & 1 & 0 & 0 \\
%1 & 0 & 1 & 0 \\
%1 & 1 & 1 & 0 \\
%1 & 0 & 0 & 1 \\
%1 & 1 & 0 & 1 \\
%1 & 0 & 1 & 1 \\
%1 & 1 & 1 & 1 \\
%\end{array}\right].
%\end{equation}
\begin{equation}
G=\left[ \begin{array}{@{}*{16}{c}@{}}
1 & 1 & 1 & 1 & 1 & 1 & 1 & 1 \\
0 & 1 & 0 & 1 & 0 & 1 & 0 & 1 \\
0 & 0 & 1 & 1 & 0 & 0 & 1 & 1 \\
0 & 0 & 0 & 0 & 1 & 1 & 1 & 1
\end{array}\right].
\end{equation}

	One can notice that the minimum distance of $\mathcal{C}_1$ is $8,$ of $\mathcal{C}_2$ is $4,$ and of $\tilde{\mathcal{C} }_3$ and $\overline{\mathcal{C}}_3$ is $2.$ Then, because of the dependence created by the main code $\mathcal{C}$ (Eq. \eqref{eq3level}), in order to calculate the squared minimum distance of $\Gamma_{C^\star},$ we may consider the combinations of codewords that yields in the minimum, i.e.,
\begin{align*}
d_{\min}^2(\Gamma_{C^\star}) = & \min \{3^2 + 7 , 2^2d_{H}(\mathcal{C}_2),  2^4d_{H}(\tilde{\mathcal{C}_3}),  \\
& 2^4d_{H}(\overline{\mathcal{C}_3})\} \\
= & \min \{16, 2^2 \cdot 4, 2^4 \cdot 2, 2^4 \cdot 2 \} =16,
%= & \min\{24, 16, 32, 32\} = 16, 
\end{align*}
where the term $3^2+7$ refers to the squared minimum distance of points that have distinct codewords in the first level.  Therefore $d_{\min}(\Gamma_{C^\star})=4.$ Here, $d_{H}$ denotes the minimum Hamming weight of the respective code. Hence, the packing density of this construction is calculated by
\vspace{-0.1cm}
\begin{align}
\Delta(\Gamma_{C^\star})=  &~\dfrac{|\mathcal{C}|~vol(\mathcal{B}_8(0, \tfrac{d_{\min}}{2}))}{2^{3n}} = \dfrac{2 \cdot 2^4 \cdot 2^{7}}{2^{24}} \dfrac{\pi^4}{4!} 2^{8}  \nonumber \\
\approx  &~ 0.25367,
\end{align}
which coincides with the packing density of the $E_8$ lattice and $E_8 = \tfrac{1}{\sqrt{8}} \Gamma_{C^\star}.$ This construction is just to illustrate that one can achieve the same packing density as $E_8$ lattice via  Construction $C^\star,$ although the most efficient way of representing this lattice is via Construction A.

\vspace{0.2cm}

\noindent \textbf{Dimension 14:} Consider $\mathcal{C}_2$ as the self-dual code $[14,7,4].$ Thus,
\begin{align*}
d_{\min}^2(\Gamma_{C^\star}) = & \min \{9+13, 2^2 \cdot 4, 2^4 \cdot 2\} =16,
\end{align*}
whose Hermite constant is
\begin{align}
\gamma_{14}(\Gamma_{C^\star}) & = \dfrac{d_{\min}^{2}(\Gamma_{C^\star})}{vol(\Gamma_{C^\star})^{2/n}} = \dfrac{16}{(2^{21})^{{2}/{14}}} = 2.
\end{align}
The upper bound for the Hermite constant in this dimension is $2.4886,$ according to \cite{Cohn2003}.

\vspace{0.2cm}

	In dimension $16,$ the best known packing density is given by the decoupled version of Eq.\eqref{constrcstar}, where $\mathcal{C}_2=\mathcal{RM}(2,4),$ where $\mathcal{RM}(r,m)$ denotes the Reed-Muller code of length $2^m$ and order $r.$ In this particular case, Construction C, D and $C^\star$ coincides. 

%\noindent \textbf{Dimension 16 - $BW_{16}$ lattice:} We can write,
%\begin{equation}
%BW_{16} = \mathcal{RM}(0,4) + 2\mathcal{RM}(2,4) + 4\mathcal{RM}(4,4) + 8\mathbb{Z}^{16},
%\end{equation}
%where $\mathcal{RM}(r,m)$ denotes the Reed-Muller code of length $2^m$ and order $r.$ %Recall that $\mathcal{RM}(0,4)$ is the $[16,1,16]-$repetition code and $\mathcal{RM}(4,4) = \mathbb{F}_2^{16}.$

\vspace{0.2cm}

\noindent \textbf{Dimension 24:} (Leech lattice) This construction was already presented in \cite{bzc2018, bzc2019} and it assumes $\mathcal{C}_2$ as the $[24,12,8]-$extended Golay code. 

\vspace{0.2cm}

\noindent \textbf{Dimension 32:} Define $\mathcal{C}_2$ as the $\mathcal{RM}(2,5),$ which is a $[32,16,8]-$ self-dual code. Then, we have that, following an analogous calculation for the minimum distance as it was done in the $E_8$ case,
\begin{align*}
d_{\min}^2(\Gamma_{C^\star}) = & \min \{9+31, 2^2 \cdot 8, 2^4 \cdot 2\} =32.
%= & \min\{48, 32, 32\} = 32.
\end{align*}
%Hence, the packing density of this construction is calculated by
%\begin{align*}
%\Delta(\Lambda_{C^\star}) & = \dfrac{|\mathcal{C}|~vol(\mathcal{B}_{32}(0, \tfrac{d_{\min}}{2}))}{2^{3n}} \approx 4.3 \times 10^{-6}
%%= \dfrac{2 \cdot 2^{16} \cdot 2^{31}}{2^{96}} \dfrac{\pi^{16}}{16!} (2\sqrt{2})^{32} \approx 4.3 \times 10^{-6}
%\end{align*}
Hence, the Hermite constant is 
\begin{align}
\gamma_{32}(\Gamma_{C^\star}) & = \dfrac{d_{\min}^{2}(\Gamma_{C^\star})}{vol(\Gamma_{C^\star})^{2/n}} = \dfrac{32}{(2^{48})^{{2}/{32}}} = 4,
\end{align}
which coincides with Hermite constant of the Barnes-Wall lattice $BW_{32}.$

%	Due to the fact that the packing density tends to zero as the dimension increase, for higher dimensions, a better measure to compare efficiency is the Hermite constant.

\vspace{0.2cm}
	
\noindent \textbf{Dimension 40:} Define $\mathcal{C}_2$ as an extremal self-dual $[40,20,8]-$code, i.e., its minimum distance achieves the highest possible value for given $k$ and $n.$ The squared minimum distance is given by
\begin{align*}
d_{\min}^2(\Gamma_{C^\star}) = & \min \{9+39, 2^2 \cdot 8, 2^4 \cdot 2\} =32.
%= & \min\{48, 32, 32\} = 32
\end{align*}
%and $d_{\min}(\Lambda_{C^\star})=4\sqrt{2}.$ 
The Hermite constant of this lattice constellation is
\begin{align}
\gamma_{40}(\Gamma_{C^\star}) & = \dfrac{d_{\min}^{2}(\Gamma_{C^\star})}{vol(\Gamma_{C^\star})^{2/n}} = \dfrac{32}{(2^{60})^{{2}/{40}}} = 4,
 \end{align} 
 which coincides the Hermite constant given by the extremal even unimodular lattice in dimension $40.$
 
\vspace{0.2cm}
 
%\vspace{0.3cm}
	
%\noindent \textbf{Dimension 64 - $BW_{64}$ lattice:} This is also a special case of Construction $C^\star$ that coincides with Construction C, where 
%\begin{equation}
%BW_{64} = \mathcal{RM}(1,6) + 2\mathcal{RM}(3,6) + 4\mathcal{RM}(5,6) + 8\mathbb{Z}^{64}.
%\end{equation}
%For this lattice, we have that
%\begin{align}
%\gamma_{64}(\Lambda_{C^\star}) & = \dfrac{d_{\min}^{2}(\Lambda_{C^\star})}{vol(\Lambda_{C^\star})^{2/n}} = \dfrac{32}{(2^{80})^{{2}/{64}}} \approx 5.6568. 
%\end{align} 

\section{Special non-lattice constellations} \label{Sec:5}

	One can notice that the scheme proposed for Construction $C^\star$ may be also used to get non-lattice constellations when the code $\mathcal{C}_2$ is not self-orthogonal or it is but does not contain the codeword $(1,\dots, 1).$
	
\vspace{0.2cm}
	
\noindent \textbf{Dimension 4:} It is believed that the best known packing density for any constellation in dimension $n=4$ is given by the lattice $D_4$ \cite{CS95, conwaysloane}, which is, up to congruence, the unique lattice that achieves this density. In the sequel, we present a non-lattice constellation that achieves the same packing density as $D_4.$% fact that also occurs with the $D_3$ lattice (face centered cubic), for example.  

%\begin{proposition}\cite[p. 390]{CS95} The only tight packing $\Pi_4$ is the root lattice $\Lambda_4 = D_4.$
%\end{proposition}

%	The notation $\Pi_4$ represents a densest $4-$dimensional packing that fibers over a tight $2^k-$dimensional packing, where $2^k$ is the highest power of $2$ strictly less than $n.$ 
	
  	We consider $\mathcal{C}_1$ and $\mathcal{C}_3$ as the coupled codes according to Section \ref{Sec:3-level}, and $\mathcal{C}_2$ is the $\mathcal{RM}(1,2)$ $[4,3,2]-$code, i.e.,
\begin{align*}
\mathcal{RM}(1,2)& = \{(0,0,0,0),(1,0,1,0),(0,1,0,1),(1,1,1,1),\\
& (0,0,1,1),(0,1,1,0),(1,0,0,1),(1,1,0,0)\},
\end{align*}
we can see that this code is not self-orthogonal. Moreover, if we apply a Construction $C^\star$ as proposed in Eq. \eqref{lattice_cstar}, it does not give a lattice. Indeed, consider $(4,6,0,2), (4,4,2,2) \in \Gamma_{\mathcal{C}^\star}.$ Their real sum is $(8,10,2,4) = (0,0,0,0)+2(0,1,1,0)+4(0,0,0,1)+8(1,1,0,0)$ and $(0,0,0,0,0,1,1,0,0,0,0,1) \notin \mathcal{C} \subseteq \mathbb{F}_2^{12}.$ 
~When we calculate the squared minimum distance of this constellation, we have that
\begin{equation*}
d_{\min}^2(\Gamma_{C^\star}) = \min \{9+3, 2^2 \cdot 2, 2^4 \cdot 2\} = \min\{12, 8, 32\} = 8
\end{equation*}
and $d_{\min}(\Gamma_{C^\star})=2\sqrt{2}.$ The packing density of this construction is then 
\begin{align*}
\Delta(\Gamma_{C^\star}) & = \dfrac{|\mathcal{C}|~vol(\mathcal{B}_4(0, \tfrac{d_{\min}}{2}))}{2^{3n}} = \dfrac{2 \cdot 2^3 \cdot 2^{3}}{2^{12}} \dfrac{\pi^2}{2!} (\sqrt{2})^{4} \\
& = \dfrac{\pi^2}{2^4} \approx 0.6168...
\end{align*}
which is the same packing density as the $D_4$ lattice.\\

	Other interesting non-lattice cases obtained by an analogous construction are the following:\\
	
\noindent \textbf{Dimension 18:} Considering $\mathcal{C}_2$ to be the $[18,9,6]-$binary linear code \cite{Simonis1992}, the resulting constellation achieves the best known Hermite constant in this dimension \cite{Cohn2003}. 

\vspace{0.2cm}

\noindent \textbf{Dimension 20:} The best sphere packing in dimension $20$ is presented in the work of Vardy \cite{Vardy95} and it can be seen as a Construction $C^\star,$ where the three levels are coupled.

\vspace{0.2cm}

\noindent \textbf{Dimension 40:} By assuming $\mathcal{C}_2$ as the $[40,23,8]-$binary linear code \cite[p. 146]{conwaysloane}, we can slightly improve the Hermite constant of the lattice presented in Section \ref{Sec:4} in dimension $40,$ which in this case reaches $\gamma_{40}= 4.287.$

\vspace{0.2cm}

%	whose construction of the best sphere packing in dimension $20$ can be seen as a Construction $C^\star$ over $\mathbb{F}_4.$

\section{Conclusion and future work} \label{Sec:Conc}

	We detailed some lattice constructions under the perspective of a special scheme of Construction $C^\star,$ using coupled first and third levels and admitting as second level self-dual codes. This construction is only interesting for low dimensions, because the choice of the most significant bit code (third level) forces an upper bound for the squared minimum distance equal to $32,$ which does not depend on the dimension. This drawback may be solved by applying Construction $C^\star$ to other families of coupled codes or by increasing the number of levels.
	
	We also presented non-lattice constructions, including a four dimensional Construction $C^\star$ that achieves the same packing density as the $D_4$ lattice and interesting potentially interesting results for dimensions $18, 20,$ and $40.$ We aim in a future work to apply other self-dual codes to Construction $C^\star,$ also with different alphabet sizes, and compare it with known results for Construction A \cite{Nebe2006}. %In this direction, we address the remarkable work from Vardy \cite{Vardy95}, whose construction of the best sphere packing in dimension $20$ can be seen as a Construction $C^\star$ over $\mathbb{F}_4.$
	
%	this multilevel construction for asymptotic dimensions, considering the use of LDPC codes or codes with a given structure, as BCH codes, for example. Our interest also extends to.
	In terms of efficient decoding, the idea is to generalize the bounded-distance decoding scheme for the Leech lattice proposed by Forney \cite{forney2} to any $3-$level lattice Construction $C^\star$ built according the structure proposed by this paper.  %with the structure of $\mathcal{C}_1$ as the repetition code, $\mathcal{C}_2$ as a self-orthogonal code and $\mathcal{C}_3$ as alternating between even and odd parity check code, according to the choice of the codeword in $\mathcal{C}_1.$
	
\vspace{0.2cm}
	
\section*{Acknowledgment}

The authors would like to thank Joseph J. Boutros for fruitful discussions and also the reviewers for meaningful suggestions. SIRC was supported by CNPq (313326/2017-7) and FAPESP (2013/25977-7) Foundations, and RZ was supported by Israel Science Foundation (676/15).

\vspace{0.2cm}


\begin{thebibliography}{1}


%\bibitem{agrelleriksson} \label{agrelleriksoon}
%E.~Agrell and T.~Eriksson, "Optimization of Lattice for Quantization". \emph{IEEE Trans. Inf. Theory}, vol. 44, no. 5, pp. 1814-1828, Sep. 1998.

%\bibitem{amrani} \label{amrani}
%O.~Amrani \emph{et al}, "The Leech Lattice and the Golay Code: Bounded-Distance Decoding and Mu1tilevel Constructions". \emph{IEEE Trans. Inf. Theory}, vol. 40, no. 4, pp. 1030-1043, Jul. 1994.

\bibitem{amrani}
O.~Amrani, Y.~Be’ery, A.~Vardy, F.-W.~Sun, and H.~C.~A.~van~Tilborg. ``The Leech lattice and the Golay code: bounded-distance decoding and mu1tilevel constructions''. \emph{IEEE Trans. on Inf. Th.}, vol. 40, no. 4, pp. 1030-1043, Jul. 1994.

\bibitem{bachoc96} 
C.~Bachoc, ``Aplications of coding theory to construction of unimodular lattices'', \textit{Jour. of Comb. Th.}, vol. 78, n. 1, pp. 92-119, Apr. 1997.

\bibitem{bzc2018} 
M.~F. Bollauf, R.~Zamir and Sueli I. R. Costa, ``Construction $C^\star:$ an inter-level coded version of Construction C'', \textit{2018 Int. Zur. Sem. on Inf. and Comm.}, Zurich, pp. 118-122, Feb.~2018.

\bibitem{bzc2019}
M.~F. Bollauf, R.~Zamir and S.I.R.~Costa, ``Multilevel constructions: coding, packing and geometric uniformity'', \textit{IEEE Trans. on Inf. Th.}, vol. 65, n. 12, pp. 7669-7681, Dec. 2019.

%\bibitem{bollaufzamir} 
%M.~F. Bollauf and R.~Zamir, "Uniformity properties of Construction C", in \emph{2016 IEEE Inter. Symp. on Inform. Theory}, (Barcelona), 2016, pp. 1516-1520.

%\bibitem{bonnecaze} 
%A.~Bonnecaze \emph{et al}, "Quaternary Quadratic Residue
%Codes and Unimodular Lattices". \emph{IEEE Trans. on Inform. Theory}, vol. 41, no. 2, pp. 366-377, Mar. 1995.

\bibitem{Cohn2003}
H.~Cohn and N.~Elkies, ``New upper bounds on sphere packings I'', \emph{Ann. of Math.}, vol. 157, no. 2, pp. 689.714, 2003.

\bibitem{Cohn2016}
H.~Cohn, A.~Kumar, S.~D.~Miller, D.~Radchenko, and M.~Viazovska,
``The sphere packing problem in dimension 24,''
%%Mar.~2016. Available at: http://arxiv.org/abs/1603.06518
{\em Ann. Math}, vol.~185, no.~3, pp.~1017-1033, Apr.~2017.

\bibitem{CS95}
J.~H. Conway and N.J.~A. Sloane, ``What are all the best sphere packings in low dimensions?'', \emph{Discr. Comput. Geom}, vol. 13, pp. 382 - 403, 1995. 

\bibitem{conwaysloane} 
J.~H. Conway and N.J.~A. Sloane, \emph{Sphere Packings, Lattices and Groups}, 3rd~ed. New York, USA: Springer, 1999.
  
\bibitem{debuda1} 
R.~de Buda, ``Fast FSK signals and their demodulation''. \emph{Can. Electron. Eng. Journal}, vol. 1, pp. 28–34, Jan. 1976.

\bibitem{ebeling}
W.~Ebeling and F.~ Hirzebruch, \emph{Lattices and codes: a course partially based on lectures by F.~ Hirzebruch}. Wiesbaden: Vieweg, 2002.
   
\bibitem{forney1} 
G.~D. Forney, ``Coset codes-part I: introduction and geometrical classification''. \emph{IEEE Trans. Inf. Theory}, vol. 34, no. 5, pp. 1123-1151. Sep. 1988.

\bibitem{forney2} 
G.~D. Forney, ``A bounded-distance decoding algorithm for the Leech lattice with generalizations''. \emph{IEEE Trans. Inf. Theory}, vol. 35, no. 4, pp. 906-909. Jul. 1989.

%\bibitem{huffpless}
%W.~C.~Huffman and V.~Pless, \emph{Fundamentals of error correcting codes}. New York, NY: Cambridge University Press, 2003.

\bibitem{Hales:2005}
T.~Hales, ``A proof of the Kepler conjecture'', \emph{Ann. Math.}, vol. 162, pp. 1065-1185. 2005.

\bibitem{Huffman:2005}
W.~C.~Huffman, ``On the classification and enumeration of self-dual codes'', \emph{Fin. F. and Their App.}, vol. 11, pp. 451-490, 2005.

\bibitem{kositoggier} 
W.~Kositwattanarerk and F.~Oggier, ``Connections between Construction D and related constructions of lattices''. \emph{Designs, Codes and Cryptography}, v. 73, pp. 441-455, Nov. 2014. 

\bibitem{MacWilliams1977}
F.~J.~MacWilliams and N.~J.~A.~Sloane, \emph{The theory of error-correcting codes}. New York, NY: North Holland Publishing Co., 1977.

\bibitem{Nebe2006}
G.~Nebe, E.~M.~Rains, and N.~J.~A.~Sloane, \emph{Self-dual codes and invariant theory}. Netherlands: Springer-Verlag Berlin, 2006.

%\bibitem{zehavi} 
%E.~Zehavi, "8-PSK trellis codes for a Rayleigh channel". \emph{IEEE Trans. Commun.}, vol. 40, no. 3, pp. 873–884, May 1992.

\bibitem{Simonis1992}
J.~Simonis, ``The [18,9,6] code is unique'', \emph{Discr. Math.}, vol. 106-107, pp. 439-448, 1992.

\bibitem{Sloane:1979}
N.~J.~A.~Sloane, ``Self-dual codes and lattices'', \emph{Proc. Symp. Pure Math.}, vol. 34, pp.273-308.

\bibitem{Vardy95}
A.~Vardy, ``A New Sphere Packing in 20-Dimensions'', \emph{Inven. Math.}, vol.~121, n0.~1, pp.~119-134, 1995.

\bibitem{Viazovska2016}
M.~Viazovska,
``The sphere packing problem in dimension 8,''
%%Mar.~2016. Available at: http://arxiv.org/abs/1603.04246
{\em Ann. Math}, vol.~185, no.~3, pp.~991-1015, Apr.~2017.

\end{thebibliography}
\end{document}